\documentclass[a4paper,10pt,reqno]{amsart}

\usepackage{hyperref}
\usepackage{amssymb}
\usepackage{eucal}
\usepackage{eufrak}
\usepackage{graphics}
\usepackage{graphicx}
\usepackage{epsfig}
\usepackage{multicol}
\usepackage{xypic}
\usepackage{amsmath}
\usepackage{wasysym}
\usepackage{dsfont}

\usepackage{eufrak}
\usepackage{graphics}
\usepackage{graphicx}
\usepackage{epsfig}
\usepackage{multicol}
\usepackage{xypic}
\usepackage{amsmath}
\usepackage{color}

\newtheorem{definition}{Definition}[section]
\newtheorem{lemma}[definition]{Lemma}
\newtheorem{remark}[definition]{Remark}

\newtheorem{theorem}[definition]{Theorem}

\newtheorem{corollary}[definition]{Corollary}

\newcommand{\lp}{\left(}
\newcommand{\rp}{\right)}
\newcommand{\lc}{\left\{}
\newcommand{\rc}{\right\}}
\newcommand{\der}{\partial}

\newcommand{\bra}{\langle}
\newcommand{\ket}{\rangle}
\newcommand{\R}{\mathbb{R}}      

\newcommand{\N}{\mathbb{N}}      
\newcommand{\T}{\mathbb{T}}

\newcommand{\Flder}{\rightarrow}

\usepackage{amssymb}

\newcommand{\proa}{A^*G \mbox{$\;$}_{\tau^*} \kern-3pt\times_\alpha
G \mbox{$\;$}_\beta \kern-3pt\times_{\tau^*} A^*G}

\newcommand{\ldb}{[\![}
\newcommand{\rdb}{]\!]}

\begin{document}

\title[A Fractional Approach for Modelling Dissipative Mechanical Systems]{A Fractional Variational Approach for Modelling Dissipative Mechanical Systems: Continuous and Discrete Settings}


\maketitle

\begin{center}
{\bf \large Fernando Jim\'enez and Sina Ober-Bl\"obaum}
\end{center}

\begin{center}
Department of Engineering Science, University of Oxford\\
 Parks Road, Oxford OXI 3PJ, UK\\
e-mail: fernando.jimenez@eng.ox.ac.uk\\
$\,\,\,\,$e-mail: sina.ober-blobaum@eng.ox.ac.uk
\end{center}

\begin{abstract}                
Employing a phase space which includes the (Riemann-Liouville) fractional derivative of curves evolving on real space, we develop a restricted variational principle  for Lagrangian systems yielding the so-called restricted fractional Euler-Lagrange equations (both in the continuous and discrete settings), which, as we show, are invariant under linear change of variables. This  principle relies on a particular restriction upon the admissible variation of the curves. In the case of the half-derivative and mechanical Lagrangians, i.e. kinetic minus potential energy, the restricted fractional Euler-Lagrange equations model a dissipative system in both directions of time, summing up to a set of equations that is invariant under time reversal. Finally, we show that the discrete equations are a meaningful discretisation of the continuous ones.
\end{abstract}

\section{Introduction}
Variational principles are powerful tools for the modelling and simulation of mechanical 
systems. As it is well-known, the fulfilment of a variational principle leads to the Euler-Lagrange equations of motion describing the dynamics of such systems. On the other hand, a variational discretisation directly yields unified numerical schemes with powerful structure-preserving properties. Using these variational or symplectic integrators (VIs, \cite{MaWe01}) (integration schemes which themselves satisfy a variational principle), the geometric structure is not destroyed and characteristic properties of the dynamics are inherited by the discrete approximation such as preservation of energy and momentum maps in the presence of symmetries. 
Continuous and discrete Lagrangian mechanics have been primarily used to model and to simulate conservative systems, 
for which the excellent long-time energy behaviour of VIs is explained by the fact that the discrete solution is the exact solution of a nearby Lagrangian system, which is a result obtained by backward error analysis (BEA) techniques (\cite{HaLuWa}).
However, because most classical processes observed in the engineering world are non-conservative, it is important to be able to apply the power of variational methods to these cases. 

There have been several attempts to provide a general method of dealing with non-conservative forces in classical mechanics (\cite{Riewe1996}; \cite{Goldstein1980}, pg. 21; \cite{CELEGHINI1989}; \cite{Bloch,Kane00,ObJuMa10}). However, for all approaches in the literature, either the physical meaning of the induced quantities is not physically menaningful or the proposed principles are not purely variational in the ususal sense. 


A promising purely variational approach was formulated in \cite{Riewe1996}: it is shown that a simple system with linear damping can be modelled using a variational principle involving fractional derivatives and a Euler-Lagrange equation of the form
\[
\sum_{\alpha} (-1)^{\alpha} \frac{d^{\alpha}}{dt^\alpha} \left( \frac{\partial  L}{\partial q^{(\alpha)}} \right)=0,
\]
where $\alpha \in \mathbb{Q}$ the set of rational numbers. 
Several definitions of a fractional derivative have been proposed (for an overview see e.g.~\cite{TheBook} and \cite{Agrawal2002}), among which we pick Riemann-Liouville's. Most importantly to mention is that the fractional derivative of a function is not determined by the value of the function at a single value $t$, but depends on the values over an entire interval.

While Riewe introduced fractional Euler-Lagrange equations for dissipative systems based on fractional Lagrangian functionals using left and right Riemann-Liouville derivatives, $D_{-}^\alpha$ and $D_{+}^\alpha$, it is shown in \cite{Cresson2012975} that critical points of a fractional Lagrangian functional and solutions of a dissipative equation are not equivalent. The reason is that due to fractional integration by parts a mixing of left and right fractional derivatives occurs and in general
\[
D_{-}^{1/2} D_{+}^{1/2} x \not= \frac{dx}{dt}.
\]
To overcome this problem, in \cite{Cresson2012975} so called \emph{asymmetric fractional embeddings} are introduced. By artificially doubling the state space and by restricting the set of variations, fractional Euler-Lagrange equations are derived whose solutions are equivalent to solutions of dissipative equations. However, the geometric structure of such enlarged phase space is not provided and the meaning of the restricted variations is not clear and apparently {\it ad hoc} chosen. 

In this paper, we also consider an enlarged fractional phase space and clarify its geometric structure in terms of vector spaces and bundle products (section \ref{FracTSpace}). We develop a restricted variational principle and obtain a sufficient condition for the extremals providing the so-called {\it restricted fractional Euler-Lagrange equations} (theorem \ref{ConsFraEL}) (see \cite{Agrawal2002} for other approaches to obtain fractional Euler-Lagrange equations). This suffcient condition relies on a restriction upon the admissible variations, allowing different varied curves for the extra variables with the same variations (definition \ref{ConsVariations}).  We show as well that the restricted fractional Euler-Lagrange equations are invariant under linear change of variables (theorem \ref{InvarianceEqs}). In the $\alpha=1/2$ case and considering mechanical Lagrangins, i.e. kinetic minus potential energy, the obtained equations describe a dissipative mechanical system, whereas the dynamics of the extra variables can be interpreted as the resversed time dynamics (corollary \ref{Coro}). As a whole, the system of equations is invariant under time reversal. In  section  \ref{Discrete}, we develop the discrete counterpart of the previous results. We obtain the {\it restricted discrete fractional Euler-Lagrange equations} (theorem \ref{DiscreteEquations}) and prove a new property for double discrete fractional derivatives (lemma \ref{LemmaDisc}) which allows to derive a meaningful discretisation of the dissipative mechanical equations in the $\alpha=1/2$ mechanical case.

\section{Riemann-Liouville fractional derivatives}

Let $\alpha\in[0,1]\subset\R$ and $f:[a,b]\subset\R\Flder\R$ a smooth function. The Riemann-Liouville fractional derivatives are defined by
\begin{subequations}\label{RL}
\begin{align}
&D^{\alpha}_{-}f(t)=\,\,\,\,\,\frac{1}{\Gamma(1-\alpha)}\frac{d}{dt}\int_a^t(t-\tau)^{-\alpha}f(\tau)d\tau,\\
&D^{\alpha}_{+}f(t)=-\frac{1}{\Gamma(1-\alpha)}\frac{d}{dt}\int_t^b(\tau-t)^{-\alpha}f(\tau)d\tau,
\end{align}
\end{subequations}
for $t\in$\,[a,b], and $\Gamma(z)$ the gamma function (\cite{TheBook}). As it is well-known, the fractional derivatives are non-local operators: in the sequel $-$ and $+$ will denote the retarded and advanced cases, respectively. Let us consider two smooth functions $f,g$. The fractional integration by parts rule is given by
\begin{equation}\label{IntegrationByParts}
\int_a^bf(t)D^{\alpha}_{\sigma}g(t)dt=\int_a^b\lp D^{\alpha}_{-\sigma}f(t)\rp g(t)dt,
\end{equation}
where $\sigma$ stands either for $-$ or $+$. An important feature of fractional integrals is, when $\alpha=1/2$:
\begin{equation}\label{DoubleFracInt} 
D^{1/2}_{-}D^{1/2}_{-}f(t)=\frac{d}{dt}f(t),\,\, D^{1/2}_{+}D^{1/2}_{+}f(t)=-\frac{d}{dt}f(t).
\end{equation}
See \cite{TheBook} for more details. According to the above definitions, the fractional derivatives are $\R$-valued. 

\section{Fractional phase space}\label{FracTSpace}


Consider a smooth curve $\gamma:[a,b]\subset\R\Flder\R^d$ for $d\in\N$. The local representation of the curve is given by $\gamma(t)=(x^1(t),...,x^d(t))$, $t\in[a,b]$. For the set of all smooth curves $C^{\infty}([a,b],\R^d)$, let us define the {\it fractional tangent vector} of the curve $\gamma$ by means of the following application
\[
\begin{split}
&X_{\sigma}^{\alpha}: C^{\infty}([a,b],\R^d) \Flder\quad\R^d,\\
&\quad\quad\quad\quad\quad\,\,\,\gamma\quad\quad\,\mapsto\,\,\,\, X_{\sigma}^{\alpha}\gamma,
\end{split}
\]
where the fractional tangent vector is defined by
\[
X_{\sigma}^{\alpha}\gamma:=D^{\alpha}_{\sigma}\gamma(t)=(D^{\alpha}_{\sigma}x^1(t),...,D^{\alpha}_{\sigma}x^d(t)),\, t\in[a,b].
\]
In the sequel we shall omit the $t$-dependence in curves and fractional tangent vectors.


\begin{definition}
 We define the {\rm fractional tangent space} of the curve $\gamma$ as
{\rm
\[
V^{\alpha}_{\sigma}\R^d=\lc X_{\sigma}^{\alpha}\gamma\,|\, \mbox{for} \,\,\gamma\in C^{\infty}([a,b],\R^d)\,,\,t\in [a,b]\rc.
\]
}
\end{definition}
Taking into account the vector structure of $C^{\infty}([a,b],\R^d)$, which is defined pointwise for $t\in[a,b]$, it is easy to prove that   $V^{\alpha}_{\sigma}\R^d$ is a $d$-dimensional vector space.

With these elements, we define the {\it fractional tangent bundle} for curves $\gamma$ as follows.

\begin{definition}\label{FracTB}
{\it Define the {\rm fractional tangent bundle} for $\gamma$ and $t\in [a,b]$ by
{\rm\[
T^{\alpha}_{\sigma}\R^d=\lc (\gamma,X^{\alpha}_{\sigma}\gamma)\in\R^d\times V^{\alpha}_{\sigma}\R^d\,|\,\gamma\in C^{\infty}([a,b],\R^d)\rc.
\]}
Whe shall consider $\R^d$ as the base space and $V^{\alpha}_{\sigma}\R^d$ the fiber.  Denoting $\mathcal{V}^{\alpha}_{\tiny{ (\gamma,\sigma)}}:=(\gamma,X_{\sigma}^{\alpha}\gamma)\in T^{\alpha}_{\sigma}\R^d$ its local representation  is given by
$\mathcal{V}^{\alpha}_{\tiny (\gamma,\sigma)}=(x^1,...,x^d, D^{\alpha}_{\sigma}x^1,..$\\$..,D^{\alpha}_{\sigma}x^d).$
The bundle projection $\tau_{\sigma}^{\alpha}:T_{\sigma}^{\alpha}\R^d\Flder\R^d$ is locally given by
\[
\tau_{\sigma}^{\alpha}(\mathcal{V}^{\alpha}_{(\gamma,\sigma)})=(x^1,...,x^d).
\]}
\end{definition}
According to this definition, we see that $(\tau_{\sigma}^{\alpha})^{-1}(\R^d)=V^{\alpha}_{\sigma}\R^d$. It is easy to see that at each time $t$ this  leads to the  local isomorphism $T^{\alpha}_{\sigma}\R^d\cong\R^d\times\R^d$.  As it is apparent, we have constructed the fractional tangent bundle in a similar way as the usual tangent bundle $T\R^d$ for curves in $\R^d$ (\cite{AbMa}).


Consider now two smooth curves $\gamma_{x}, \gamma_{y}:[a,b]\subset\R\Flder\R^d$, forming the new curve $\tilde\gamma:[a,b]\subset\R\Flder\R^d\times\R^d$ by $\tilde\gamma=(\gamma_{x},\gamma_{y})$. Locally it is expressed by
$
\tilde\gamma=(x^1,...,x^d,y^1,...,y^d).
$
\begin{definition}\label{DoubleFracTB}
Define the {\rm double fractional tangent bundle} for curves $\tilde\gamma$ and $t\in[a,b ]$ by
{\rm\[
\begin{split}
\T^{\alpha}\R^d=\{\lp(\gamma_{x},\gamma_{y}), (X^{\alpha}_{-}\gamma_{x},X^{\alpha}_{+}\gamma_{y})\rp
&\in (\R^d\times\R^d)\times (V^{\alpha}_{-}\R^d\times V^{\alpha}_{+}\R^d)\,\\&\mbox{such that}
\quad\gamma_{x},\gamma_{y}\in C^{\infty}([a,b],\R^d)\},
\end{split}
\]}
where now $\R^d\times\R^d$ is the base space and $V^{\alpha}_{-}\R^d\times V^{\alpha}_{+}\R^d$ the fiber. Denoting $\mathcal{V}^{\alpha}_{\tilde\gamma}:=\lp(\gamma_{x},\gamma_{y}), (X^{\alpha}_{-}\gamma_{x},X^{\alpha}_{+}\gamma_{y})\rp\in\T^{\alpha}\R^d$, its local representation is given by
$
\mathcal{V}^{\alpha}_{\tilde\gamma}=(x^i,y^i,D^{\alpha}_{-}x^i,D^{\alpha}_{+}y^i),$ $i=1,...,d$;
while the bundle projection $\tau^{\alpha}: \T^{\alpha}\R^d\Flder \R^d\times\R^d$ is 
\[
\tau^{\alpha}(\mathcal{V}^{\alpha}_{\tilde\gamma})=(x,y).
\]
\end{definition}
From now on we shall drop the $i$ superindex for sake of simplicity.
\begin{remark}
The vector structure of $(\tau^{\alpha})^{-1}(\R^d\times\R^d)=V^{\alpha}_{-}\R^d\times V^{\alpha}_{+}\R^d$ follows directly from the vector structure of each $V^{\alpha}_{\sigma}\R^d$ and its Cartesian product. Moreover, we have that {\rm dim}$(V^{\alpha}_{-}\R^d\times V^{\alpha}_{+}\R^d)=2d$. Locally we have the isomorphism
\[
\T^{\alpha}\R^d\cong T^{\alpha}_{-}\R^d\times  T^{\alpha}_{+}\R^d\cong \R^d\times\R^d\times\R^d\times\R^d.
\]
\end{remark}
Analogously to definition \ref{DoubleFracTB}, we can define the {\it double tangent bundle} by
\[
\begin{split}
\T\R^d=\{\lp(\gamma_{x},\gamma_{y}), (X_{-}\gamma_{x},X_{+}\gamma_{y})\rp\in &(\R^d\times\R^d)\times (V_{-}\R^d\times V_{+}\R^d)\,\\&\mbox{such that}\quad
\gamma_{x},\gamma_{y}\in C^{\infty}([a,b],\R^d)\},
\end{split}
\]
where now we are considering the usual tangent bundle of $\R^d$ instead of the fractional one. For a curve $\gamma\in C^{\infty}([a,b],\R^d)$, a vector $X_{\sigma}\gamma\in V_{\sigma}\R^d$ shall be  defined by its time derivative, i.e.  $X_{\sigma}\gamma:=\dot\gamma$. Defining $\mathcal{V}_{\tilde\gamma}:=\lp(\gamma_{x},\gamma_{y}), (X_{-}\gamma_{x},X_{+}\gamma_{y})\rp$, we have the local expression
\[
\mathcal{V}_{\tilde\gamma}=(x,y,\dot x,\dot y).
\]
The bundle projection is $\tau:\T\R^d\Flder\R^d\times\R^d$, $\tau(\mathcal{V}_{\tilde\gamma})=(x,y)$.
\begin{remark}
{\it
 $V_{\sigma}\R^d$  and $V_{-}\R^d\times V_{+}\R^d$ are vector spaces. Moreover, locally
\[
\T\R^d\cong T_{-}\R^d\times  T_{+}\R^d\cong \R^d\times\R^d\times\R^d\times\R^d.
\]
}
\end{remark}


In the forthcoming sections we are going to consider Lagrangian functions defined on a particular phase space that we describe now. 
\begin{definition}\label{FracPhSp}
{\it Consider the double fractional tangent bundle $\T^{\alpha}\R^d$ and the double tangent bundle $\T\R^d$. We define the {\rm fractional phase space} as the bundle product of them over $\R^d\times\R^d$, i.e.
\[
\mathfrak{T}\R^d:=\T\R^d\otimes_{\tiny{\R^d\times\R^d}}\T^{\alpha}\R^d.
\]
Thus, $\mathcal{V}_{\tilde\gamma}:=(\gamma_{x},\gamma_{y},X_{-}\gamma_{x}, X_{+}\gamma_{y},X_{-}^{\alpha}\gamma_{x}, X_{+}^{\alpha}\gamma_{y})\in\mathfrak{T}\R^d$ is locally described by
\begin{equation}\label{TotalCoord}
\mathcal{V}_{\tilde\gamma}=(x,y,\dot x,\dot y,D^{\alpha}_{-}x,D^{\alpha}_{+}y).
\end{equation}
The bundle projection $\mathcal{T}:\mathfrak{T}\R^d\Flder\R^d\times\R^d$ is defined by $\mathcal{T}(\mathcal{V}_{\tilde\gamma})=(x,y)$. }
\end{definition}
For more details on bundle products we refer to \cite{AbMa}. According to the previous development, we conclude the following (the proofs are straightforward): the fiber $\mathcal{T}^{-1}(\R^d\times\R^d)$ is a vector space with dimension $4d$; and $\mathfrak{T}\R^d$ is locally the Cartesian product of 6 copies of $\R^d$.



\section{Restricted continuous variational principle}\label{ContVarPrin}

Consider a smooth curve  $\tilde\gamma=(\gamma_x,\gamma_y)$, for $\gamma_x,\gamma_y\in C^{\infty}([a,b],\R^d)$ with fixed endpoints $x_a=\gamma_{x}(a)=\gamma_{y}(b)$ and $x_b=\gamma_{x}(b)=\gamma_{y}(a)$ in $\R^d$. 
\begin{definition}
{\it Define the path space by
\[
\begin{split}
C^{\infty}(x_a,x_b;\R^d\times\R^d)=\{\tilde\gamma\in C^{\infty}([a,b],\R^d\times\R^d)\,|\tilde\gamma(a)=(x_a,x_b),\, \tilde\gamma(b)=(x_b,x_a)\}.
\end{split}
\]}
\end{definition}

\begin{definition}\label{Variations}
{\it Given a curve $\tilde\gamma\in C^{\infty}(x_a,x_b;\R^d\times\R^d)$, a {\rm  varied curve} is a map $\Gamma:\R\times[a,b]\Flder\R^d\times\R^d$,
\[
\Gamma(\epsilon,t):=\tilde\gamma(t)+\epsilon\delta\tilde\gamma(t)
\]
for $\delta\tilde\gamma:[a,b]\Flder\R^d\times\R^d$ a {\rm  variation}, which is smooth and defined such that $\delta\tilde\gamma(a)=\delta\tilde\gamma(b)=0$.}
\end{definition}

\begin{remark}
{\it Note that  we are establishing that $\gamma_{y}(a)=x_b$ and $\gamma_{y}(b)=x_a$. This will make sense afterwards since we shall interpret $\gamma_y$ as $\gamma_{x}$ for reversed time (see Remark~\ref{rem:interpret_rev}).}
\end{remark}
From definition \ref{Variations}  it follows that $\delta\gamma_{x}(a)=\delta\gamma_{y}(b)=\delta\gamma_{x}(a)=\delta\gamma_{y}(b)=0$. Locally
\begin{equation}\label{Endpoints}
\delta x(a)=\delta y(a)=\delta x(b)=\delta y(b)=0.
\end{equation}
The variations $\delta\tilde\gamma$ induce the map $\delta\widehat{\tilde\gamma}:[a,b]\Flder\mathfrak{T}\R^d$. Using coordinates \eqref{TotalCoord} and defintion \ref{Variations}, it follows that at each $t\in[a,b]$
\begin{equation}\label{HatGamma}
\delta\widehat{\tilde\gamma}=(\delta x,\,\,\delta y,\,\,\delta\dot x,\,\, \delta\dot y,\,\, D^{\alpha}_{-}x,\,D^{\alpha}_{+}y).
\end{equation}
Next, we define the set of restricted varied curves.

\begin{definition}\label{ConsVariations}
{\it Define the set of {\rm restricted varied curves} as $\Gamma_{\eta}(\epsilon,t):=\tilde\gamma(t)+\epsilon\eta(t)$, where $\delta\tilde\gamma=\eta(t)=(\delta\gamma_x(t),\delta\gamma_x(t))$. In other words, we impose $\delta\gamma_x(t)=\delta\gamma_y(t)$ to the variations, which locally is expressed by $\delta x=\delta y.$}
\end{definition}

\begin{remark}
{\it
Observe that under definition \ref{ConsVariations} we are allowing different varied curves, i.e $\tilde\gamma_x\neq\tilde\gamma_y$, while restricting the variations, i.e. $\delta\gamma_x=\delta\gamma_y$.
}
\end{remark}

Now,  define a $C^2$ Lagrangian function $L:\mathfrak{T}\R^d\Flder\R$ and the action functional $\mathcal{A}:C^{\infty}(x_a,x_b;\R^d\times\R^d)\Flder\R$ given by
\begin{equation}\label{ActFunc}
\mathcal{A}(\tilde\gamma):=\int_a^bL(\mathcal{V}_{\tilde\gamma})\,dt.
\end{equation}
We have already all the ingredients to establish a restricted variational principle:

\begin{theorem}\label{ConsFraEL}
{\it A curve $\tilde\gamma:[a,b]\Flder \R^d\times\R^d$, subject to restricted variations $\eta(t)$ in definition \ref{ConsVariations}, is an extremal of the action $\mathcal{A}:C^{\infty}(x_a,x_b,\R^d\times\R^d)\Flder\R$ defined in \eqref{ActFunc} if it satisfies the {\rm restricted fractional  Euler-Lagrange equations}:
\begin{subequations}\label{FracEL}
\begin{align}
\frac{\der L}{\der x^i}-\frac{d}{dt}\lp\frac{\der L}{\der \dot x^i}\rp+D^{\alpha}_{-}\lp\frac{\der L}{\der D^{\alpha}_{+} y^i}\rp=0,\label{FracEL:-}\\
\frac{\der L}{\der y^i}-\frac{d}{dt}\lp\frac{\der L}{\der \dot y^i}\rp+D^{\alpha}_{+}\lp\frac{\der L}{\der D^{\alpha}_{-} x^i}\rp=0,\label{FracEL:+}
\end{align}
\end{subequations}
for $i=1,...,d.$}
\end{theorem}
\begin{proof} To find the extremals of $\mathcal{A}$ for restricted varied curves $\Gamma_{\eta}(\epsilon,t)$ we impose the usual critical condition, i.e. $\delta\mathcal{A}:=\frac{d}{d\epsilon}\mathcal{A}(\Gamma_{\eta})\big|_{\epsilon=0}=0$. Using the expression \eqref{HatGamma} we obtain:
\[
\begin{split}
\delta\mathcal{A}=&\int_a^b\bra\nabla L(\mathcal{V}_{\tilde\gamma}),\delta\widehat{\tilde\gamma}\ket\,dt=\int_a^b\Big\{\frac{\der L}{\der x}\delta x+\frac{\der L}{\der\dot x}\delta\dot x+\frac{\der L}{\der D_{-}^{\alpha}x} D_{-}^{\alpha}\delta x\\
+&\frac{\der L}{\der y}\delta x+\frac{\der L}{\der\dot y}\delta\dot x+\frac{\der L}{\der D_{+}^{\alpha}y}D_{+}^{\alpha}\delta x\Big\}\,dt\\
=&\int_{a}^b\Big\{\frac{\der L}{\der x}-\frac{d}{dt}\lp\frac{\der L}{\der \dot x}\rp+D^{\alpha}_{+}\lp\frac{\der L}{\der D^{\alpha}_{-} x}\rp\\
+&\frac{\der L}{\der y}-\frac{d}{dt}\lp\frac{\der L}{\der \dot y}\rp+D^{\alpha}_{-}\lp\frac{\der L}{\der D^{\alpha}_{+} y}\rp\Big\}\delta x\,dt+\frac{\der L}{\der \dot x}\delta x\Big|_a^b+\frac{\der L}{\der \dot y}\delta x\Big|_a^b,
\end{split}
\]
where in the second equality we have employed the constraints $\delta x=\delta y$, while in the third equality we have used integration by parts with respect to the total and fractional derivatives \eqref{IntegrationByParts}. According to the endpoint conditions \eqref{Endpoints}, all the border terms vanish leading to
\begin{eqnarray}
\delta\mathcal{A}=&&\int_{a}^b\Big\{\frac{\der L}{\der x}-\frac{d}{dt}\lp\frac{\der L}{\der \dot x}\rp+D^{\alpha}_{+}\lp\frac{\der L}{\der D^{\alpha}_{-} x}\rp\nonumber\\
&&\quad\,\,+\frac{\der L}{\der y}-\frac{d}{dt}\lp\frac{\der L}{\der \dot y}\rp+D^{\alpha}_{-}\lp\frac{\der L}{\der D^{\alpha}_{+} y}\rp\Big\}\delta x\,dt\nonumber\\
=&&\int_{a}^b\Big\{\Big[\frac{\der L}{\der x}-\frac{d}{dt}\lp\frac{\der L}{\der \dot x}\rp+D^{\alpha}_{-}\lp\frac{\der L}{\der D^{\alpha}_{+}y}\rp\Big]\,\delta x\label{RestVar}\\
&&\quad\,\,\,+\Big[\frac{\der L}{\der y}-\frac{d}{dt}\lp\frac{\der L}{\der \dot y}\rp+D^{\alpha}_{+}\lp\frac{\der L}{\der D^{\alpha}_{-}x}\rp\Big]\delta x\Big\}\,dt.\nonumber
\end{eqnarray}
Finally, from this last expression of $\delta\mathcal{A}$ is easy to see that the restricted fractional Euler-Lagrange equations \eqref{FracEL} are a sufficient condition for $\delta\mathcal{A}=0$; and the claim holds. 
\end{proof}

As mentioned above, equations \eqref{FracEL} are only sufficient conditions for the extremal curves.

\begin{theorem}\label{InvarianceEqs}
{\it The restricted fractional  Euler-Lagrange equations \eqref{FracEL} are invariant under linear change of variables $x=\Lambda\,z$ and  $y=\Lambda\,\tilde z$, where $\Lambda\in \mathbb{M}^{d\times d}(\R)$ is full rank.}
\end{theorem}
\begin{proof} Since the fractional derivative operators \eqref{RL} are linear (so are total time derivatives), we have that
\[
\begin{split}
\dot x=\Lambda\,\dot z,\quad\quad D^{\alpha}_{-}x=\Lambda\,D^{\alpha}_{-}z,\\
\dot y=\Lambda\,\dot{\tilde z},\quad\quad D^{\alpha}_{+}y=\Lambda\,D^{\alpha}_{+}\tilde z.
\end{split}
\]
Taking this into account and using the chain rule, it is easy to prove that $\delta\mathcal{A}$ in its \eqref{RestVar} form can be rewritten as
\[
\begin{split}
&\delta\mathcal{A}=\int_{a}^b\Big\{\Big[\frac{\der L}{\der z}-\frac{d}{dt}\lp\frac{\der L}{\der z}\rp+D^{\alpha}_{-}\lp\frac{\der L}{\der D^{\alpha}_{+}\tilde z}\rp\Big]\Lambda^{-1}\delta x\\
&\quad\quad\quad\,\,\,\,+\Big[\frac{\der L}{\der \tilde z}-\frac{d}{dt}\lp\frac{\der L}{\der \dot{\tilde z}}\rp+D^{\alpha}_{+}\lp\frac{\der L}{\der D^{\alpha}_{-}z}\rp\Big]\Lambda^{-1}\delta x\Big\}\,dt,
\end{split}
\]
where we have used that $\Lambda$ is full rank. Again $\delta x$ is arbitrary, and therefore a sufficient condition for the extremal curves is
\[
\begin{split}
\Big[\frac{\der L}{\der z}-\frac{d}{dt}\lp\frac{\der L}{\der \dot z}\rp+D^{\alpha}_{-}\lp\frac{\der L}{\der D^{\alpha}_{+} \tilde z}\rp\Big]\Lambda^{-1}=0,\\
\Big[\frac{\der L}{\der \tilde z}-\frac{d}{dt}\lp\frac{\der L}{\der \dot{\tilde z}}\rp+D^{\alpha}_{+}\lp\frac{\der L}{\der D^{\alpha}_{-} z}\rp\Big]\Lambda^{-1}=0.
\end{split}
\]
Using again that $\Lambda$ is full rank, we obtain the restricted fractional Euler-Lagrange equations \eqref{FracEL} in the new coordinates $(z,\tilde z,\dot z,\dot{\tilde z},D^{\alpha}_{-}z,D^{\alpha}_{+}\tilde z)$ and the claim holds. 
\end{proof}

Define now the Lagrangian $L(x,y,\dot x,\dot y,D^{\alpha}_{-}x,D^{\alpha}_{+}y)$ by
\begin{equation}\label{PartLagrangian}
L(x,y,\dot x,\dot y,D^{\alpha}_{-}x,D^{\alpha}_{+}y):=L_x(x,\dot x)+L_y(y,\dot y)-\ldb D^{\alpha}_{-}x,D^{\alpha}_{+}y\rdb_{_{R}},
\end{equation}
where $L_x:T_{-}\R^d\Flder\R$  and $L_y:T_{+}\R^d\Flder\R$ are $C^2$ functions, and $\ldb\cdot,\cdot\rdb_{_{R}}:V^{\alpha}_{-}\R^d\times V^{\alpha}_{+}\R^d\Flder\R$ is a symmetric bilinear form defined by
\begin{equation}\label{WedgeProd}
\begin{split}
\ldb D^{\alpha}_{-}x, D^{\alpha}_{+}y\rdb_{_{R}}:=(D^{\alpha}_{-}x)^TR\,D^{\alpha}_{+}y,
\end{split}
\end{equation}
where $R=$diag$(\rho_1,...,\rho_d)\in\mathbb{M}^{d\times d}(\R^+)$. With this new Lagrangian we obtain the following corollary to theorem \ref{ConsFraEL}.
\begin{corollary}\label{Coro}
{\it If $L$ is given by \eqref{PartLagrangian}, $\alpha=1/2$ and $L_x(x,\dot x)=\frac{1}{2}\dot x^TM\,\dot x-U(x)$, $L_y(y,\dot y)=\frac{1}{2}\dot y^TM\,\dot y-U(y)$, where $M=${\rm diag}$(m_1,...,m_d)\in\mathbb{M}^{d\times d}(\R^+)$ and $U:\R^d\Flder\R$ is a smooth function, then \eqref{FracEL} are
\begin{subequations}\label{Damping}
\begin{align}
&M\,\ddot x+R\,\dot x+\nabla U(x)=0,\\
&M\,\ddot y-R\,\dot y+\nabla U(y)=0.
\end{align}
\end{subequations}}
\end{corollary}

\begin{proof}
The proof is straightforward after plugging \eqref{PartLagrangian} into \eqref{FracEL},
\[
\begin{split}
&\frac{\der L_x}{\der x}-\frac{d}{dt}\lp\frac{\der L_x}{\der \dot x}\rp-D^{\alpha}_{-}D_{-}^{\alpha}x=0\\
&\frac{\der L_y}{\der y}-\frac{d}{dt}\lp\frac{\der L_y}{\der \dot y}\rp-D^{\alpha}_{+}D^{\alpha}_{+}y=0.
\end{split}
\]
Moreover, 
replacing the particular Lagrangians $L_x,L_y$ in the claim and using  \eqref{DoubleFracInt}
 when $\alpha=1/2$, we arrive at equations \eqref{Damping}. 
\end{proof}

\begin{remark}\label{rem:interpret_rev}
We observe that if we interpret that $\gamma_y(t)$ is $\gamma_x(t)$ in reversed time, i.e. $\gamma_y(t)=\gamma_x(a+b-t)$, the system of second order differential equations \eqref{Damping} is invariant under time reversal, this is $t\Flder a+b-t$.
\end{remark}

\begin{remark}\label{EnergyIntegral}
 Define as usual the total energy at time $t$ $E(t):=E(x,\dot x)+E(y,\dot y)=\frac{1}{2}\dot x^TM\,\dot x+U(x)+\frac{1}{2}\dot y^TM\,\dot y+U(y)$. One can check that     $\frac{d}{dt}\lp E(t)-E(a+b-t)\rp=0$ under the dynamics \eqref{Damping} and the assumption $\gamma_y(t)=\gamma_x(a+b-t)$ . Consequently, the quantity $E(t)-E(a+b-t)$ becomes an {\rm energy-related} first integral of the dissipative system \eqref{Damping}.
\end{remark}

\begin{remark}
As observed above, fractional derivatives are non-local. In particular,  $D^{\alpha}_{+}y(t)$ depends on ``future'' times in the interval $(t,b]$. This violates causality of physical laws and therefore we would not be allowed to introduce such terms in the action in order to obtain the dynamical equations of a given system. However, the interpretation of $y(t)$ as $x(a+b-t)$ (Remark \ref{rem:interpret_rev}) helps overcome this issue, since ``future'' becomes ``past'' in reversed time, and consequently $D^{\alpha}_{+}y(t)$ respects causality.
\end{remark}

\section{Restricted discrete  variational principle}\label{Discrete}
Let us consider the increasing sequence of times $\{ t_k=hk\,|\,k=0,...,N\}\subset\R$ where $h$ is the fixed time step  determined by $h=(b-a)/N$. Define a discrete curve as a collection of points in $\R^d$ i.e. $\gamma_d:=\lc x_0,x_1,...,x_{N-1},x_N\rc=\lc x_k\rc_{0:N}\in\R^{(N+1)d}$. As  usual, we will consider these points as an approximation of the continuous curve at time $t_k$, namely $x_k\simeq x(t_k).$ Given $\lc z_k\rc_{0:N}$ (later on we shall particularise in $\lc x_k\rc_{0:N}$ and $\lc y_k\rc_{0:N}$) define the following sequences:
\begin{subequations}\label{AllSequencs}
\begin{align}
&\lc S\,z_k\rc_{\tiny 0:N-1},\,\,\,\quad\,\,\, S\,z_k:=\frac{z_k+z_{k+1}}{2},\label{MidpointRule}\\
&\lc \Delta_{-}z_k\rc_{\tiny 1:N},\,\,\,\,\,\,\,\quad\Delta_{-}z_k:=\frac{z_{k}-z_{k-1}}{h},\,\Delta_{-}z_0:=0, \label{DiscDerDef:1}\\
&\lc \Delta_{+}z_k\rc_{\tiny 0:N-1},\,\,\quad\Delta_{+}z_k:=-\frac{z_{k+1}-z_{k}}{h},\label{DiscDerDef:2}\\
&\lc \Delta_{-}^{\alpha}z_k\rc_{\tiny 0:N},\,\,\,\,\,\,\,\,\quad\Delta_{-}^{\alpha}z_k:=\frac{1}{h^{\alpha}}\sum_{n=0}^k\alpha_nz_{k-n},\label{DiscFracDerDef:1}\\
&\lc \Delta_{+}^{\alpha}z_k\rc_{\tiny 0:N},\,\,\,\,\,\,\,\,\quad\Delta_{+}^{\alpha}z_k:=\frac{1}{h^{\alpha}}\sum_{n=0}^{N-k}\alpha_nz_{k+n},\label{DiscFracDerDef:2}
\end{align}
\end{subequations}
where

\begin{equation}\label{AlphaDef}
\alpha_n:=\frac{-\alpha\,(1-\alpha)\,(2-\alpha)\cdot\cdot\cdot(n-1-\alpha)}{n!};\quad\alpha_0:=1.
\end{equation}
Naturally, we shall consider $\Delta_{-}x_k$ (resp. $\Delta_{+}y_k$) as an approximation of $\dot x(t_k)$ (resp. $\dot y(t_k)$) and $\Delta_{-}^{\alpha}x_k$ (resp. $\Delta_{+}^{\alpha}y_k$) as an approximation of $D^{\alpha}_{-} x(t_k)$ (resp. $D^{\alpha}_{+} y(t_k)$).  For more details on the approximation of fractional derivatives we refer to (\cite{CressonBook}, Chapter 5).
\begin{remark}
We observe that  \eqref{DiscFracDerDef:1}, \eqref{DiscFracDerDef:2} are well-defined for $k=0$ and $k=N$. Namely, straightforward computations lead to $\Delta_{-}^{\alpha}x_0=\alpha_0x_0/h^{\alpha}$ and $\Delta_{+}^{\alpha}y_N=\alpha_0y_N/h^{\alpha}$.
\end{remark}
Given two sequences $\lc F_k\rc_{\tiny 0:N}, \lc G_k\rc_{\tiny 0:N}$, the discrete derivatives \eqref{DiscDerDef:1}, \eqref{DiscDerDef:2}, and discrete fractional derivatives \eqref{DiscFracDerDef:1}, \eqref{DiscFracDerDef:1}, obey the following discrete integration by parts relationships:
\[
\begin{split}
\sum_{k=0}^{N-1}F_k(\Delta_{+}G_k)&=\sum_{k=1}^N(\Delta_{-}F_k)G_k+\frac{1}{h}F_0G_0-\frac{1}{h}F_NG_N,\\
\sum_{k=0}^{N-1}F_k(\Delta_{+}^{\alpha}G_k)&=\sum_{k=1}^N(\Delta_{-}^{\alpha}F_k)G_k+\frac{1}{h^{\alpha}}F_0G_0-\frac{1}{h^{\alpha}}F_NG_N.
\end{split}
\]
See \cite{CressonBook} for proofs. Define now the sets of discrete curves
\[
\begin{split}
C_d^x&=\lc \gamma_d^x=\lc x_k\rc_{\tiny 0:N}\in\R^{(N+1)d}\,|\,x_0=x_a,\,\,x_N=x_a\rc,\\
C_d^y&=\lc \gamma_d^y=\lc y_k\rc_{\tiny 0:N}\in\R^{(N+1)d}\,|\,\,y_0=x_b,\,\,\,y_N=x_a\rc,
\end{split}
\]
and the action sum $\mathcal{A}_d:C_d^x\times C_d^y\Flder\R$,
\begin{equation}\label{DiscActSum}
\mathcal{A}_d(\gamma_d^x,\gamma_d^y):=h\sum_{k=0}^{N-1}L(S\,x_k,S\,y_k,\Delta_{-}x_{k},\Delta_{+}y_{k},\Delta_{-}^{\alpha}x_k,\Delta_{+}^{\alpha}y_k),
\end{equation}
where we consider a Lagrangian function $L:\mathfrak{T}\R^d\Flder\R$ as defined in section \ref{ContVarPrin}. We observe that the evaluation of such a Lagrangian function at $(S\,x_k,S\,y_k,\Delta_{-}x_{k},\Delta_{+}y_{k},$\\$\Delta_{-}^{\alpha}x_k,\Delta_{+}^{\alpha}y_k)$ for sequences \eqref{AllSequencs} makes sense since $\mathfrak{T}\R^d\cong\R^d\times\R^d\times\R^d\times\R^d\times\R^d\times\R^d$. We define now the variation of discrete curves.
\begin{definition}\label{DiscreteVariations}
{\it Given a discrete curve $\tilde\gamma_d=(\gamma_d^x,\gamma_d^y)\in C_d^x\times C_d^y$, we define the set of {\rm  varied discrete curves} by
\[
\Gamma^x_{\epsilon}:=\gamma_d^x+\epsilon\,\delta\gamma^x_d,\quad
\Gamma^y_{\epsilon}:=\gamma_d^y+\epsilon\,\delta\gamma^y_d.
\]
where $\delta\gamma_d^x:=\lc \delta x_k\rc_{0:N}$, $\delta\gamma_d^y:=\lc \delta y_k\rc_{0:N}$ are the {\rm discrete variations}, defined such that 
\begin{equation}\label{DiscEndpoints}
\delta x_0=\delta x_N=0,\quad\,\delta y_0=\delta y_N=0.
\end{equation}
We define the set of {\rm restricted  varied discrete curves}, by
\[
\Gamma^x_{(\epsilon,\eta)}:=\gamma_d^x+\epsilon\,\eta_d,\quad
\Gamma^y_{(\epsilon,\eta)}:=\gamma_d^y+\epsilon\,\eta_d,
\]
where we establish $\eta_d=\delta\gamma_d^x=\delta\gamma_d^y$. Locally this means that $\delta x_k=\delta y_k$ for $k=1,...,N-1$.}
\end{definition}

In the following, we shall set 
\[
L(k):=L(S\,x_k,S\,y_k,\Delta_{-}x_{k},\Delta_{+}y_{k},\Delta_{-}^{\alpha}x_k,\Delta_{+}^{\alpha}y_k)
\]
for simplicity.
\begin{theorem}\label{DiscreteEquations}
{\it A discrete curve $\tilde\gamma_d=(\gamma_d^x,\gamma_d^y)\in C_d^x\times C_d^y$, subject to restricted discrete variations $\eta_d$ in definition \ref{DiscreteVariations}, is an extremal of the action sum $\mathcal{A}_d$ defined in \eqref{DiscActSum} if it satisfies the {\rm restricted discrete fractional Euler-Lagrange equations}:
\begin{subequations}\label{DiscFracEL}
\begin{align}
&\frac{1}{2}\lp D_1L(k)+D_1L(k-1)\rp+\Delta_{+}D_3L(k)+\Delta_{-}^{\alpha}D_6L(k)=0,\\
&\frac{1}{2}\lp D_2L(k)+D_2L(k-1)\rp\,+\Delta_{-}D_4L(k)+\Delta_{+}^{\alpha}D_5L(k)=0,
\end{align}
\end{subequations}
for  $k=1,...,N-1.$ $D_mL:=\frac{\der L}{\der z_m^i}$, where $z_m$ is the $m$-th variable in \eqref{DiscActSum}, m=1,\ldots,6, and $i=1,\ldots,d.$}
\end{theorem}
\begin{proof} The extremals will satisfy the condition $\delta\mathcal{A}_d:=\frac{d}{d\epsilon}\mathcal{A}_d(\Gamma^x_{(\eta,\epsilon)},\Gamma^y_{(\eta,\epsilon)})\big|_{\epsilon=0}=0.$
\[
\begin{split}
&\delta\mathcal{A}_d=h\sum_{k=0}^{N-1}\big\{D_1L(k)S\,\delta x_k+D_2L(k)S\,\delta x_k+D_3L(k)\Delta_{-}\delta x_{k}\\
&\quad\quad\quad\quad+D_4L(k)\Delta_{+}\delta x_k+D_5L(k)\Delta_{-}^{\alpha}\delta x_k+D_6L(k)\Delta_{+}^{\alpha}\delta x_k\big\}
\end{split}
\]
where we have imposed the restricted variations $\delta x_k=\delta y_k$ and have taken into account that $\delta \lp S\,x_k\rp=S\,\delta x_k$, $\delta (\Delta_{-} x_{k})=\Delta_{-}\delta x_{k}$, $\delta (\Delta_{+} x_{k})=\Delta_{+}\delta x_{k},$  $\delta (\Delta_{-}^{\alpha}x_{k})=\Delta_{-}^{\alpha}\delta x_{k}$, $\delta (\Delta_{+}^{\alpha}x_{k})=\Delta_{+}^{\alpha}\delta x_{k}$ according to \eqref{AllSequencs} and  \eqref{DiscEndpoints}. We display the treatment of terms 1,3,5 (for terms 2,4,6 we use equivalent techniques).
\[
\begin{split}
\sum_{k=0}^{N-1}D_1L(k)S\,\delta x_k&=\sum_{k=0}^{N-1}\frac{1}{2}D_1L(k)\lp\delta x_k+\delta x_{k+1}\rp\\
&=\sum_{k=1}^{N-1}\frac{1}{2}\lp D_1L(k)+D_1L(k-1)\rp\delta x_k;
\end{split}
\]
where in the second equality we have rearranged the summation index and taken into account the endpoint conditions $\delta x_0=\delta x_N=0.$
\[
\begin{split}
&\sum_{k=0}^{N-1}D_3L(k)\Delta_{-}\delta x_{k}=\sum_{k=1}^{N}D_3L(k)\Delta_{-}\delta x_{k}\\
&\quad=\sum_{k=0}^{N-1}\Delta_{+}D_3L(k)\delta x_{k}=\sum_{k=1}^{N-1}\Delta_{+}D_3L(k)\delta x_{k};
\end{split}
\]
in the first equality we have used $\delta x_0=\delta x_N=0$ and \eqref{DiscDerDef:1}, in the second equality employed the discrete integration by parts  and, in the third,  $\delta x_0=0.$ 
\[
\begin{split}
&\sum_{k=0}^{N-1}D_5L(k)\Delta_{-}^{\alpha}\delta x_{k}=\sum_{k=1}^{N}D_5L(k)\Delta_{-}^{\alpha}\delta x_{k}\\
&\quad\quad=\sum_{k=0}^{N-1}\Delta_{+}^{\alpha}D_5L(k)\delta x_{k}=\sum_{k=1}^{N-1}\Delta_{+}^{\alpha}D_5L(k)\delta x_{k};
\end{split}
\]
in the first equality we have employed $\delta x_0=\delta x_N=0$, in the second used the discrete fractional integration by parts, and in the third $\delta x_0=0$. Gathering all terms together and ordering them conveniently we arrive to
\[
\begin{split}
&\delta\mathcal{A}_d=h\sum_{k=1}^{N-1}\Big(\Big\{\frac{1}{2}\lp D_1L(k)+D_1L(k-1)\rp\\
&\quad\quad\quad+\Delta_{+}D_3L(k)+\Delta_{-}^{\alpha}D_6L(k)\Big\}\delta x_k\\
&\quad\quad\quad\quad\quad+\Big\{\frac{1}{2}\lp D_2L(k)+D_2L(k-1)\rp\\
&\quad\quad\quad\quad\quad\quad\quad+\Delta_{-}D_4L(k)+\,\Delta_{+}^{\alpha}D_5L(k)\Big\}\delta x_k\Big).
\end{split}
\]
Taking into account that $\delta x_k$ are arbitrary for $k=1,...,N-1$ we observe that  the discrete fractional Euler-Lagrange equations \eqref{DiscFracEL} are sufficient conditions for the extremal discrete curves. 
\end{proof}

Let us consider the Lagrangian function \eqref{PartLagrangian} and $L_x,L_y$ given in corollary \ref{Coro}, i.e.
\begin{equation}\label{DiscLagFin}
\begin{split}
L(k)=\frac{1}{2}&(\Delta_{-}x_{k})^TM\,\Delta_{-}x_{k}+\frac{1}{2}(\Delta_{+}y_{k})^TM\,\Delta_{+}y_{k}\\
&\quad\quad\quad-U(S\,x_k)-U(S\,y_k)-\ldb\Delta_{-}^{\alpha}x_k, \Delta_{+}^{\alpha}y_k\rdb_{_{R}}.
\end{split}
\end{equation}
Straightforward computations provide:
\[
\begin{split}
&D_1L(k)=-\nabla U(S\,x_k),\,\,\, D_3L(k)=M\Delta_{-}x_k,\,\,\,\,D_5L(k)=-R\Delta^{\alpha}_{+}y_k,\\
&D_2L(k)=-\nabla U(S\,y_k),\,\,\,\, D_4L(k)=M\Delta_{+}y_k,\,\,\,\,D_6L(k)=-R\Delta^{\alpha}_{-}x_k.
\end{split}
\]
Plugging these expressions into \eqref{DiscFracEL} we obtain
\[
\begin{split}
-\frac{1}{2}\nabla U(S\,x_k)-\frac{1}{2}\nabla U(S\,x_{k-1})+M\Delta_{+}\Delta_{-}x_{k}-R\Delta_{-}^{\alpha}\Delta_{-}^{\alpha}x_k=0,\\
-\frac{1}{2}\nabla U(S\,y_k)-\frac{1}{2}\nabla U(S\,y_{k-1})+M\Delta_{-}\Delta_{+}y_k-R\Delta_{+}^{\alpha}\Delta_{+}^{\alpha}y_k=0,
\end{split}
\]
for $k=1,...,N-1$. Employing \eqref{AllSequencs}, changing sings and reordering terms these equations are
\begin{subequations}\label{DiscEqs}
\begin{align}
&M\frac{x_{k+1}-2x_k+x_{k-1}}{h^2}\nonumber\\
&\quad\quad+R\frac{1}{h^{2\alpha}}\sum_{n=0}^k\alpha_n\sum_{p=0}^{k-n}\alpha_px_{k-n-p}\\
&\quad\quad\quad\quad\quad+\frac{1}{2}\nabla U\lp\frac{x_{k+1}+x_k}{2}\rp+\frac{1}{2}\nabla U\lp\frac{x_{k}+x_{k-1}}{2}\rp=0,\nonumber\\
&M\frac{y_{k+1}-2y_k+y_{k-1}}{h^2}\nonumber\\
&\quad\quad+R\frac{1}{h^{2\alpha}}\sum_{n=0}^{N-k}\alpha_n\sum_{p=0}^{N-k-n}\alpha_py_{k+n+p}\\
&\quad\quad\quad\quad\quad+\frac{1}{2}\nabla U\lp\frac{y_{k+1}+y_k}{2}\rp+\frac{1}{2}\nabla U\lp\frac{y_{k}+y_{k-1}}{2}\rp=0.\nonumber
\end{align}
\end{subequations}
The following lemma will help interpret equations \eqref{DiscEqs}.
\begin{lemma}\label{LemmaDisc}
{\it For $\alpha=1/2$ and $k=1,...,N-1$,}
\[
\sum_{n=0}^k\alpha_n\sum_{p=0}^{k-n}\alpha_px_{k-n-p}=x_k-x_{k-1}.
\]
\end{lemma}
\begin{proof} According to \eqref{AlphaDef} we have that $\alpha_0=1$ and $\alpha_1=-1/2$ for $\alpha=1/2$, which leads, after expanding the summations, to
\[
\begin{split}
\sum_{n=0}^k\alpha_n\sum_{p=0}^{k-n}&\alpha_px_{k-n-p}=x_k-x_{k-1}\\
&+\sum_{n=2}^{k}2\alpha_nx_{k-n}+\sum_{n=1}^k\alpha_n\sum_{p=1}^{k-n}\alpha_px_{k-n-p}.
\end{split}
\]
The claim automatically holds for $k=1$. For $k\geq 2$, arranging the sum indices we see that the previous expression can be rewritten as
\begin{eqnarray}
&&\sum_{n=0}^k\alpha_n\sum_{p=0}^{k-n}\alpha_px_{k-n-p}=x_k-x_{k-1}\nonumber\\
&&\quad\quad\quad+\sum_{s=2}^{r}\beta_0^sx_{k-s}+\sum_{l=0}^{k-(r+1)}\beta_l^{r+1}x_{k-(r+1)-n}\label{SumExp}\\
&&\quad\quad\quad\quad\quad\quad\quad\quad\quad\quad\quad+\sum_{n=r}^k\alpha_n\sum_{p=1}^{k-n}\alpha_px_{k-n-p},\nonumber
\end{eqnarray}
where for a fixed $k$ we set $r=k-1$ and 
\begin{equation}\label{Betas}
\beta_l^j=2\alpha_{l+j}+\sum_{i=1}^{j-1}\alpha_i\alpha_{l+j-i},
\end{equation}
(it is apparent that $j$ is not a power but a superindex). For a fixed $k=\tilde k$, \eqref{SumExp} acquires the form
\[
\begin{split}
&\sum_{n=0}^{\tilde k}\alpha_n\sum_{p=0}^{\tilde k-n}\alpha_px_{\tilde k-n-p}=x_{\tilde k}-x_{\tilde k-1}\\
&+\beta_0^2x_{\tilde k-2}+\beta_0^3x_{\tilde k-3}+\cdot\cdot\cdot+\beta_0^{\tilde k-2}x_2+\beta_0^{\tilde k-1}x_1+\beta_0^{\tilde k}x_0.
\end{split}
\]
According to this, it is enough to prove that $\beta_0^j=0$ for any $j$, for which we proceed by induction. From \eqref{Betas} and \eqref{AlphaDef}, it follows that $\beta^2_0=2\alpha_2+\alpha_1\alpha_1$, which vanishes for $\alpha=1/2$. Taking this as the first induction step, it is enough to prove that $\beta^{j+1}_0=0$ assuming that $\beta^j_0=0$.
\[
\begin{split}
&\beta^{j+1}_0=2\alpha_{j+1}+\sum_{i=1}^{j}\alpha_i\alpha_{j+1-i}=2\alpha_{j+1}+\sum_{i=1}^{r-1}\alpha_i\alpha_{r-i}\\
&=2\alpha_{j+1}-2\alpha_r+2\alpha_r+\sum_{i=1}^{r-1}\alpha_i\alpha_{r-i}=2\alpha_{j+1}-2\alpha_r=0,
\end{split}
\]
where we have set $r=j+1.$ Hence the claim follows.
\end{proof}

Using similar arguments, one can prove that
\[
\sum_{n=0}^{N-k}\alpha_n\sum_{p=0}^{N-k-n}\alpha_py_{k+n+p}=-(y_{k+1}-y_k),
\]
for $k=1,...,N-1$. Therefore, in the case $\alpha=1/2$ we observe that \eqref{DiscEqs} becomes
\[
\begin{split}
&M\frac{x_{k+1}-2x_k+x_{k-1}}{h^2}+R\frac{x_{k}-x_{k-1}}{h}\\
&\quad\quad+\frac{1}{2}\nabla U\lp\frac{x_{k+1}+x_k}{2}\rp+\frac{1}{2}\nabla U\lp\frac{x_{k}+x_{k-1}}{2}\rp=0,\\
&M\frac{y_{k+1}-2y_k+y_{k-1}}{h^2}-R\frac{y_{k+1}-y_k}{h}\\
&\quad\quad+\frac{1}{2}\nabla U\lp\frac{y_{k+1}+y_k}{2}\rp+\frac{1}{2}\nabla U\lp\frac{y_{k}+y_{k-1}}{2}\rp=0.
\end{split}
\]
These equations provide a classical midpoint rule variational integrator (see e.g.~\cite{MaWe01}) for the conservative part of equations \eqref{Damping}, i.e.~for kinetic and potential energy terms, whereas the velocity-dependent damping term is approximated by backward and forward difference operators for the equations in forward and backward time, respectively. Equivalent equations may be obtained from the discretisation of the Lagrange-d'Alembert principle (see \cite{MaWe01,ObJuMa10}). However, this principle implies equating the variaton of the action to the work done by the external forces, which is not the point of our approach since we obtain such external forces from the restricted variational principle.

\section{Conclusions}
In summary, we have obtained the so-called {\it restricted fractional Euler-Lagrange equations} (both continuous and discrete) from a restricted variational principle defined on a $\alpha-$fractional phase space and curves evolving on  real space. In the case of $\alpha=1/2$ and mechanical Lagrangian functions, i.e. kinetic minus potential energy, these equations model a dissipative mechanical system, and present time symmetry and invariance under linear change of variables. As one could expect, the discrete equations are a meaningful discretisation of the continuous ones, giving rise to the construction and application of variational integrators for non-conservative mechanical systems.

This work establishes solid grounds for future developments. We are in position to construct a canonical transformation providing meaningful fractional Hamilton equations. Moreover, these Hamilton equations would facilitate the introduction of controlled external forces and the development of a {\it Fractional Pontryagin Maximum Principle} leading to necessary optimality conditions for dissipative mechanical systems. 
Needless to say, the numerical behaviour of the obtained schemes shall be tested by means of proper implementations, as well as compared to state-of-the-art numerical integrators for dissipative systems.
\medskip\medskip

{\bf Acknowledgments}: This work has been funded by the EPSRC project: `'Fractional Variational Integration and Optimal Control''; ref: EP/P020402/1.

\end{document}